\documentclass[letterpaper, 10 pt, conference]{ieeeconf}  

\IEEEoverridecommandlockouts                              
\usepackage{cite}
\usepackage[utf8]{inputenc}
\usepackage{graphics} 
\usepackage{amsmath} 
\usepackage{amssymb}  
\let\labelindent\relax
\usepackage{enumitem}

\usepackage{amsthm}
\usepackage{amsmath,amssymb,amsfonts}
\usepackage{algpseudocode}
\usepackage{algorithm}
\usepackage{subcaption}
\usepackage{graphicx}  
\usepackage{caption}    
\usepackage{subcaption} 
\usepackage{booktabs}
\usepackage{siunitx}
\usepackage{makecell}
\usepackage{threeparttable}
\usepackage{enumitem} 
\newtheorem{definition}{Definition}
\newtheorem{problem}{Problem}
\newtheorem{assumption}{Assumption}
\newtheorem{proposition}{Proposition}
\newtheorem{theorem}{Theorem}

\newtheorem{lemma}{Lemma}
\newtheorem{remark}{Remark}
\usepackage{tikz}
\usetikzlibrary{positioning}
\usetikzlibrary{arrows.meta}
\usepackage{subcaption} 
\usepackage{graphics} 
\usepackage{hyperref}
\newtheorem{test}{Test}
\usepackage{xcolor} 

\title{\LARGE \bf
Hyperlyve: Hyperplane Partitioning for Neural Lyapunov Verification 
\thanks{The authors are with the Elmore Family School of Electrical Engineering at Purdue University. E-mail: \{jwayment, yarbroub, wang6203, sundara2, philpare\}@purdue.edu. 
    This material is based upon work supported in part by the US National Science Foundation (NSF-ECCS \#2238388, NSF CNS-2333487).\\
Code available at 
\url{https://github.com/Hyperlyve/hyperlyve}}}

\author{Jesse Wayment, Brian Yarbrough, Jingbo Wang, Shreyas Sundaram, and Philip E. Par\'{e}}
\begin{document}

\maketitle
\thispagestyle{empty}
\pagestyle{empty}

\begin{abstract}
This work introduces Hyperlyve (Hyperplane Partitioned Lyapunov Verifier), a novel algorithm for sound and complete verification of neural Lyapunov candidates by interpreting shallow ReLU networks as hyperplane arrangements. This perspective reduces positive definiteness verification to a finite set of vertex evaluations, and the decrease condition to a bounded optimization problem over each region. We formally prove correctness of the proposed verification procedures and demonstrate that Hyperlyve achieves significant speedups over state-of-the-art methods.
\end{abstract}

\section{INTRODUCTION}

Lyapunov function construction is a standard way of proving asymptotic stability for nonlinear dynamical systems. The existence of a Lyapunov function further allows the discovery of other important characteristics of a system, such as invariant sets and regions of attraction. However, finding a Lyapunov function for a given dynamical system is nontrivial, as there is no universal analytic method. 

Recently, neural networks have been widely implemented to find Lyapunov functions \cite{chang2019neural,gaby2022lyapunov,grune2020computing,li2025two}. Due to the Universal Approximation Theorem, it is known that a single-hidden-layer neural network can approximate any continuous function to any desired degree \cite{RELU_UAT}. However, this approach has created a need for tools for neural network verification, i.e., tools that can determine whether a given neural network is a true Lyapunov function for a given system.

Neural network verification has been well-studied, especially for problems involving reachability analysis \cite{tran2019star,xiang2018reachability}. These methods are generally employed in robustness analysis of neural network controllers.
For the problem of verifying neural Lyapunov functions, prior work  focuses on certifying the \emph{Lyapunov conditions} of a learned candidate (zero at the origin, positive definiteness, and negative Lie derivative) rather than merely checking generic neural network properties. Earlier methods mainly relied on \emph{SMT}~\cite{liu2024tool, liu_learning_max_lyapunov, huang_MNLF}, \emph{SOS/SDP}~\cite{liu2024compositionally}, or \emph{MIP/MILP} encodings~\cite{chen2021learning,dai2020counter}, which are sound but often difficult to scale because they must reason jointly about the certificate, the controller, and the nonlinear dynamics. More recent work uses $\alpha,\beta$-CROWN~\cite{yang2024lyapunov,li2025two} to accelerate this process through bound propagation and branch-and-bound on the composed Lyapunov verification problem. However, its fast mode is sound but not complete, while its exhaustive branch-and-bound mode is sound and complete but can be prohibitively slow. 
%
Therefore, there remains a need for an \emph{efficient}, \emph{sound}, and \emph{complete}\footnote{Soundness means that if a model is verified, then it is guaranteed to satisfy the target specification, while completeness means that if a model is falsified, then it is guaranteed to violate the specification.} verification method for neural Lyapunov functions.

The key observation underlying our method, called Hyperlyve (pronounced ``Hyper-Live''), is that a single-layer neural network with ReLU activation functions induces a hyperplane arrangement that partitions the state space (as illustrated in Figure~\ref{fig:hidden_neurons}). ReLU hyperplane arrangements have been used previously for reachability analysis~\cite{ferlez_reachability}. While \cite{SamanipourTAC, SamanipouraAllerton} use these properties to synthesize and verify Lyapunov functions for piecewise dynamical systems, our approach generalizes to any smooth dynamical system.

Our contribution is a novel method that leverages the properties of hyperplane arrangements to determine whether a candidate neural network is Lyapunov over a user-defined region of interest. We derive and formally prove the validity of our proposed algorithm for verifying each Lyapunov condition. Numerical experiments demonstrate that our approach outperforms state-of-the-art neural network verification methods across several benchmarks. 

The organization of this paper is as follows. In Section~\ref{section: background} we introduce the setup of our problem and give a formal problem statement. In Section~\ref{section: verification}, we explain our proposed verification method in-depth and provide several theoretical results. We explain our implementation in Section~\ref{section: our_alg}, and run simulations and experiments in Section~\ref{section: simulations} against other recent methods. After an analysis of the results, we then conclude with Section~\ref{section: conclusion}.




\subsection{Notation}

Let $\mathbb{R}^n$ denote Euclidean space and $e_i$ the standard basis vectors. We denote $\mathbf{v}_i(\mathbf{x})$ as the $i$th element of a vector-valued function $\mathbf{v}(\mathbf{x})$. We also use $\hat{0}$ to denote the maximal element in a partially ordered set, which in our case corresponds to the ambient space $\mathbb{R}^p$. We use diag$(\cdot)$ to denote a diagonal matrix, $A^T$ to denote the transpose of a matrix $A$, and  $u(\cdot)$ to represent the unit step function. Also, $[n]$ denotes the set $\{1, 2, \dots, n\}$. The ReLU activation function is defined as $\sigma(x) = \text{max}(0, x)$. We denote the cartesian product $p$ times as $I^p$ for some interval $I$. The cardinality of a set is denoted as $|\cdot|$. Let $\overline{I}$ denote the closure of set~$I$. The $\mathrm{vert}(\cdot)$ function returns the set of vertices of given a region.  

\begin{figure}[htbp]
    \centering
    \resizebox{\columnwidth}{!}{
    \begin{tikzpicture}[
        neuron/.style={circle, draw, minimum size=1.4cm, inner sep=0pt, align=center, font=\small},
        input/.style={neuron, fill=blue!5},
        hidden/.style={neuron, fill=green!5},
        output/.style={neuron, fill=red!5},
        connection/.style={->, >=stealth, thin, gray!60}
    ]

        \node[input] (I-1) at (0,0) {$x_1$};
        \node[input] (I-2) at (0,-1.8) {$x_2$}; 
        \node at (0,-2.8) {$\vdots$}; 
        \node[input] (I-p) at (0,-4) {$x_p$};
        \node[above=0.2cm of I-1] {\textbf{Input ($p$)}};

        \node[hidden] (H-1) at (3,0.8) {\scalebox{0.6}{$\sigma(\mathbf{w}_1 \cdot \mathbf{x} + b_1)$}};
        \node[hidden] (H-2) at (3,-0.8) {\scalebox{0.6}{$\sigma(\mathbf{w}_2 \cdot \mathbf{x} + b_2)$}};
        \node[hidden] (H-3) at (3,-2.4) {\scalebox{0.6}{$\sigma(\mathbf{w}_3 \cdot \mathbf{x} + b_3)$}};
        \node at (3,-3.4) {$\vdots$}; 
        \node[hidden] (H-n) at (3,-4.5) {\scalebox{0.6}{$\sigma(\mathbf{w}_n \cdot \mathbf{x} + b_n)$}};
        \node[above=0.2cm of H-1] {\textbf{Hidden ($n$)}};

        \node[output] (O) at (6,-1.85) {\scalebox{0.8}{$V_{nn}(\mathbf{x})$}};
        \node[above=0.2cm of O] {\textbf{Output}};

        \foreach \i in {1,2,p}
            \foreach \h in {1,2,3,n}
                \draw[connection] (I-\i) -- (H-\h);

        \foreach \h in {1,2,3,n}
            \draw[connection] (H-\h) -- (O);

        \draw[connection] (I-1) -- (H-1) node [pos=0.2, above, sloped, black, font=\tiny] {$w_{1}$};
        \draw[connection] (I-1) -- (H-2) node [pos=0.2, above, sloped, black, font=\tiny] {$w_{2}$};
        \draw[connection] (I-1) -- (H-n) node [pos=0.1, above, sloped, black, font=\tiny] {$w_{n}$};
        \draw[connection] (I-p) -- (H-1) node [pos=0.2, above, sloped, black, font=\tiny] {$w_{(p-1)n + 1}$};
        \draw[connection] (I-p) -- (H-n) node [pos=0.2, below, sloped, black, font=\tiny] {$w_{pn}$};
        \draw[connection] (H-1) -- (O) node [pos=0.2, above, sloped, black, font=\tiny] {$w_{pn + 1}$};
        \draw[connection] (H-n) -- (O) node [pos=0.2, above, sloped, black, font=\tiny] {$w_{(p+1)n}$};

    \end{tikzpicture}
    } 
    
    \caption{A fully connected neural network with one hidden layer. There are $(p+1)n$ weights and $n+1$ biases. We define the $p \times 1$ vector as $\mathbf{w}_l = \begin{bmatrix}
        w_l & w_{l + n} & \dots & w_{l + (p-1)n}
    \end{bmatrix}^T$ and $\mathbf{x} = \begin{bmatrix}
        x_1 & x_2 \dots & x_p
    \end{bmatrix}^T$. }
    \label{fig:neural_network}
\end{figure}
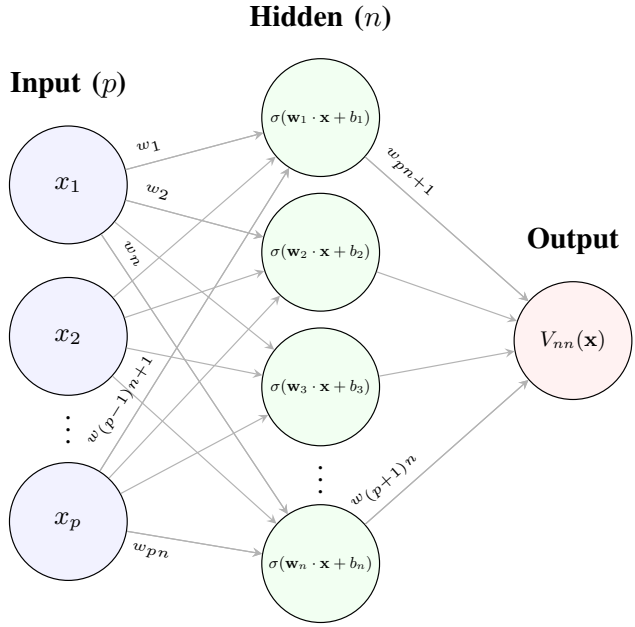

\section{BACKGROUND}
\label{section: background}
We give the following definitions as supporting information for our theoretical results. These definitions have been adapted from \cite{stanley2007introduction}.
\begin{definition}[Arrangement of Hyperplanes]
Let $H_l = \mathbf{w}_l^T \mathbf{x} + b_l = 0$ represent a hyperplane in $\mathbb{R}^p$ where $\mathbf{w}_l, \mathbf{x} \in \mathbb{R}^p$ and $b_l \in \mathbb{R}$ for $l \in [n]$. Then $\mathcal{A} = \{H_1, H_2, \dots, H_n\}$ is an arrangement of $n$ hyperplanes. 
\label{def: arrangement}
\end{definition}

\begin{definition} [Region of a Hyperplane Arrangement]
    Given a hyperplane arrangement $\mathcal{A}$ in $\mathbb{R}^p$, a region $\mathcal{U}$ in $\mathbb{R}^p$ is defined as a maximal connected subset of $\mathbb{R}^p / \bigcup_{H \in \mathcal{A}} H$. Each $\mathcal{U}_t $, for $t \in [r(\mathcal{A})]$, is a convex region which could be bounded or unbounded. The number of regions produced by $\mathcal{A}$ is denoted as $r(\mathcal{A})$.
    \label{def: region}
\end{definition}

\subsection{Problem Statement}
Consider the $p$-dimensional continuous-time dynamical system
\begin{equation}
        \dot{\mathbf{x}} = f(\mathbf{x})  = \begin{bmatrix}
            f_1(\mathbf{x}) \\
            f_2(\mathbf{x}) \\
            \vdots \\
            f_p(\mathbf{x})
        \end{bmatrix},
        \label{eq: original_system}
\end{equation}
where $\mathbf{x} \in \mathbb{R}^p$ and $f:\mathbb{R}^p \to \mathbb{R}^p$ is continuous and defines the system dynamics.

The following definitions are used in the problem statement. 

\begin{definition}[Region of Interest] \label{def:roi_B}
Consider a hyperplane arrangement $\mathcal{A}$. We define the {region of interest}, denoted $\mathcal{B}$, as $\mathcal{B} = \bigcup_{t \in r(\mathcal{A})} (\mathcal{U}_t \cap \mathcal{K})$, for a compact polytope $\mathcal{K} \subset \mathbb{R}^p$. Then, $\mathcal{B}$ is a union of bounded convex polytope regions, $\overline{\mathcal{R}}_i$, for $i \in [r_{\mathcal{B}}(\mathcal{A})]$, where $r_{\mathcal{B}}(\mathcal{A})$ represents the number of these regions in $\mathcal{B}$. 
    
\end{definition}
\begin{figure}[t]
    \centering
    \includegraphics[width=0.8\linewidth]{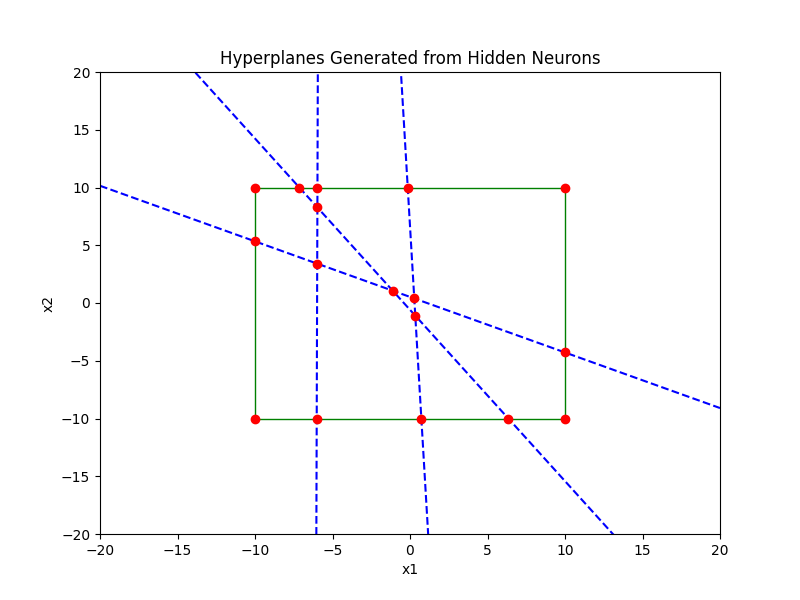}
    \caption{Generated hyperplanes from a neural network with four hidden neurons (blue). The green square is the region of interest $\mathcal{B}$, and the red dots are intersection points.}
    \label{fig:hidden_neurons}
\end{figure}
\begin{definition}[Local Lyapunov Function] \label{def:lyapunov} Consider the dynamical system in \eqref{eq: original_system}, with the origin as an equilibrium, and a region of interest $\mathcal{B} \subset \mathbb{R}^p$ with $\{0\}\subset \mathcal{B}$.
A local Lyapunov function is a scalar-valued function $V: \mathcal{B} \to \mathbb{R}_{\geq0}$ that satisfies the following properties: \begin{enumerate}[label=(\roman*)] \item $V(\mathbf{0}) = 0$ \label{cond_1}
\item $V(\mathbf{x}) > 0 $ for all $\mathbf{x} \in \mathcal{B}\setminus\{\mathbf{0}\}$ \label{cond_2} \item $\dot{V}(\mathbf{x}) = \nabla V(\mathbf{x}) \cdot f(\mathbf{x}) < 0$ for all $\mathbf{x} \in \mathcal{B}\setminus\{\mathbf{0}\}$. \label{cond_3} \end{enumerate} \end{definition}
%
    %
 Consider a single-hidden-layer ReLU neural network, $V_{nn}$, and a region of interest $\mathcal{B}$. We seek to solve the following problem. 
\begin{problem}
    Develop a method to determine if a candidate neural Lyapunov function meets Lyapunov conditions \ref{cond_1}, \ref{cond_2}, and \ref{cond_3} inside $\mathcal{B}$.  \label{prob: 1}
\end{problem}
Due to the many convenient properties of hyperplane arrangements, we present a hyperplane arrangement-based method as a solution to Problem~\ref{prob: 1}.
\begin{remark}[Piecewise-affine structure and closed regions]
Since $V_{nn}$ is a shallow ReLU network, it is affine on each region induced by the hyperplane arrangement $\mathcal{A}$, but it may fail to be differentiable on the hyperplanes in $\mathcal{A}$. Accordingly, throughout Sections~\ref{section: verification} and~\ref{section: our_alg}, statements involving $\nabla V_{nn}$ and $\dot V_{nn}$ are interpreted on
\[
\mathcal{B}\setminus \bigcup_{H\in\mathcal{A}} H.
\]
For each region $\mathcal{R}_i \subset \mathcal{B}$, let $\overline{\mathcal{R}}_i$ denote its closure relative to $\mathcal{K}$. Since $\mathcal{K}$ is a compact polytope and each region is cut out by affine halfspaces, $\overline{\mathcal{R}}_i$ is a compact convex polytope. In the positivity test below, vertices and extrema are taken over $\overline{\mathcal{R}}_i$.
\end{remark}




\section{Verification of Candidate Lyapunov Function}
\label{section: verification}
Consider a single-hidden-layer neural network with ReLU activation functions, $V_{nn}$, as a candidate Lyapunov function; see Figure~\ref{fig:neural_network}. The network has $p$ input neurons, $n$ hidden neurons, and a single output neuron. We define $w_k$ and $ b_r$ as the corresponding weights and biases of $V_{nn}$, where $k \in [(p+1)n]$, \mbox{$r \in [n+1]$}. 
Therefore, the output of $V_{nn}$ can be written as
\begin{equation}
\begin{aligned} \label{eq: V_nn-eq}
    V_{nn}(\mathbf{x}) ={} & w_{pn + 1} \sigma(\mathbf{w}_1 \cdot \mathbf{x} + b_1) \\
    & + w_{pn + 2} \sigma(\mathbf{w}_2 \cdot \mathbf{x} + b_2) + \cdots \\
    & + w_{pn + n} \sigma(\mathbf{w}_n \cdot \mathbf{x} + b_n) + b_{n+1},
\end{aligned}
\end{equation}
where, for each hidden neuron $l \in[n]$, we define the incoming weights as
\begin{equation}
    \mathbf{w}_l = \begin{bmatrix}
        w_l & w_{l + n} & \cdots & w_{l + (p-1)n}
    \end{bmatrix}^T. 
\end{equation}
The derivative of $V_{nn}$, with respect to each input $j\in [p]$, is 
\begin{equation} \label{eq: p_deriv}
\begin{aligned}
    \frac{\partial V_{nn}(\mathbf{x})}{\partial x_j} = {} & w_{pn+1} (\mathbf{w}_{1,j}) \sigma'(\mathbf{w}_1 \cdot \mathbf{x} + b_1) \\
    & + w_{pn+2} (\mathbf{w}_{2,j}) \sigma'(\mathbf{w}_2 \cdot \mathbf{x} + b_2) \\
    & + \cdots \\
    & + w_{pn+n} (\mathbf{w}_{n,j}) \sigma'(\mathbf{w}_n \cdot \mathbf{x} + b_n),
\end{aligned}
\end{equation}
where $\mathbf{w}_{l,j}$ represents the $j$th element of $\mathbf{w}_l$. 
We can rewrite this expression in matrix form:
\begin{equation}
    \frac{\partial V_{nn}(\mathbf{x})}{\partial x_j} = \mathbf{U}(\mathbf{x}) \cdot \mathbf{W} \tilde{\mathbf{w}}_j,
\end{equation}
where \begin{equation}
    \mathbf{U}(\mathbf{x}) = \begin{bmatrix}
        u(\mathbf{w}_1 \cdot \mathbf{x} + b_1) \\
        u(\mathbf{w}_2 \cdot \mathbf{x} + b_2) \\
        \vdots \\
        u(\mathbf{w}_n \cdot \mathbf{x} + b_n)
    \end{bmatrix}
    \label{eq: u_hype}
\end{equation}
is a binary vector containing only 1's or 0's; 
\begin{equation*}
    \mathbf{W} = \operatorname{diag}(w_{pn + 1}, \dots, w_{(p +1)n})
\end{equation*}
is a diagonal matrix of the weights from each hidden neuron to the output neuron;
and 
\begin{equation*}
    \tilde{\mathbf{w}}_j = \begin{bmatrix} w_{1 + (j-1)n} & w_{2 + (j - 1)n} & \dots & w_{n + (j - 1)n} \end{bmatrix}^T
\end{equation*}
is the vector of all the weights connected to the $j$th input neuron. 

Then, the gradient of $V_{nn}$ is
\begin{equation}
    \nabla V_{nn}(\mathbf{x}) = \begin{bmatrix}
        \mathbf{U}(\mathbf{x}) \cdot \mathbf{W} \tilde{\mathbf{w}}_1 \\ 
        \mathbf{U}(\mathbf{x}) \cdot \mathbf{W} \tilde{\mathbf{w}}_2
        \\
        \vdots \\ 
        \mathbf{U}(\mathbf{x}) \cdot \mathbf{W} \tilde{\mathbf{w}}_p
    \end{bmatrix}.
    \label{eq:nn_gradient}
\end{equation}

Examining~\eqref{eq: u_hype}, each expression $\mathbf{w}_l \cdot \mathbf{x} + b_l = 0$ defines a hyperplane with normal vector $\mathbf{w}_l$. Let $H_l$ denote the hyperplane associated with the $l$th hidden neuron. The set
\begin{equation}
\mathcal{A} = \{H_1,H_2,\dots,H_n\} \label{eq: hype_set}
\end{equation}
forms a hyperplane arrangement that partitions the state space into convex polytope regions. We are primarily interested in $\mathcal{B}$ and its corresponding the convex polytopes $\mathcal{R}_i$ (see Definition~\ref{def: region}). 

\begin{assumption}
    Assume we are given a  single-hidden-layer neural network with ReLU activation functions, as defined in \eqref{eq: V_nn-eq}, as a candidate Lyapunov function  for the dynamical system~\eqref{eq: original_system}. Further, assume that the region of interest, $\mathcal{B}$, is compact.
    \label{assump: 1}
\end{assumption}

Note that by Assumption~\ref{assump: 1}, the ReLU structure of the candidate induces a hyperplane arrangement as shown in $\eqref{eq: u_hype}$ and $\eqref{eq: hype_set}$. Then, $\mathcal{B}$ is composed of convex polytopes, denoted $\mathcal{R}_i$. In the following, we exploit this geometric structure to derive tests for each of the Lyapunov conditions in Definition~\ref{def:lyapunov}. 
\begin{proposition} 
     Suppose Assumption~\ref{assump: 1} holds. Let $\mathcal{R}_i$ be an arbitrary region created by the hyperplane arrangement. Then for all $\mathbf{x}, \mathbf{y} \in \mathcal{R}_i$, \mbox{$\nabla V_{nn}(\mathbf{x}) = \nabla V_{nn}(\mathbf{y})$}. 
    \label{constant_grad}
\end{proposition}

\begin{proof} 
Fix a region $\mathcal{R}_i$. By definition of a region of the arrangement, for every hidden neuron $l \in [n]$, the sign of $\mathbf{w}_l^\top \mathbf{x} + b_l$ is constant for all $\mathbf{x} \in \mathcal{R}_i$. Hence the activation pattern vector $\mathbf{U}(\mathbf{x})$ from \eqref{eq: u_hype} is constant on $\mathcal{R}_i$. Since $\mathbf{W}$ and each $\tilde{\mathbf{w}}_j$ are constant, \eqref{eq:nn_gradient} implies that $\nabla V_{nn}(\mathbf{x})$ is constant on $\mathcal{R}_i$. Therefore, for all $\mathbf{x},\mathbf{y}\in\mathcal{R}_i$,
\[
\nabla V_{nn}(\mathbf{x}) = \nabla V_{nn}(\mathbf{y}).
\]

\end{proof}
We now present how the properties of the hyperplane arrangement induced by the hidden neurons lead to a verification method for candidate neural Lyapunov functions \mbox{(i.e., verifying all three conditions in Definition~\ref{def:lyapunov}).}

Verifying Condition~\ref{cond_1} is a trivial test at the origin:

\begin{test}[Zero at Origin — Condition~\ref{cond_1}]
    \label{test:cond_1}
    Evaluate $V_{nn}(\mathbf{0})$. If $V_{nn}(\mathbf{0}) \neq 0$, $\{\mathbf{0}\}$
    is a counterexample. Otherwise, Condition~\ref{cond_1} is satisfied.
\end{test}

\subsection{Verifying Positive Definiteness (Condition~\ref{cond_2})}

 This section will introduce how Hyperlyve can verify the positivity of a candidate neural Lyapunov function in a compact region of interest $\mathcal{B}$.

\begin{proposition}
    Suppose Assumption~\ref{assump: 1} holds. Consider an arbitrary region $\mathcal{R}_i$. If there exists $\mathbf{x}_0 \in \mathcal{R}_i$ such that $V_{nn}(\mathbf{x}_0) \leq 0$, then there exists a vertex $\mathbf{a}$ of $\overline{\mathcal{R}}_i$ such that $V_{nn}(\mathbf{a}) \leq 0$.
\end{proposition}
\begin{proof}
    We proceed by contradiction. Suppose that there exists some $\mathbf{x}_0 \in \mathcal{R}_i$ such that $V_{nn}(\mathbf{x}_0) < 0$. Further, suppose that $V_{nn}$ is positive at all vertices. 

    Recall from Proposition~\ref{constant_grad} that the gradient is fixed throughout $\mathcal{R}_i$. By the continuity of ReLUs, $V_{nn}$ is affine in $\overline{\mathcal{R}}_i$. Then, by the Fundamental Theorem of Linear Programming \cite{dantzig2016linear}, it is known that a minimum occurs at some vertex~$\mathbf{a}$. Therefore, $\mathbf{x}_0$ is larger than the minimum, and~$V_{nn}(\mathbf{a})$ must also be negative. 
\end{proof}



A direct consequence of the above proposition is that if a region $\mathcal{R}_i$ contains a point that violates Condition~\ref{cond_2}, then a vertex of the closed region $\overline{\mathcal{R}}_i$ must also violate Condition~\ref{cond_2}. Thus, positivity can be verified by checking the sign of $V_{nn}$ at every vertex of every closed region. 
The vertices of each region are simple to find, as each vertex represents the intersection points of the hyperplanes in $\mathcal{A}$ or intersections with the region of interest (as seen in Figure~\ref{fig:hidden_neurons}).
Thus, Condition~\ref{cond_2} in Definition~\ref{def:lyapunov} is verified using the following test:


\begin{test}[Positive Definiteness — Condition~\ref{cond_2}]
    \label{test:cond_2}
    If there exists a vertex $\mathbf{a} \in \bigcup_i \mathrm{vert}(\overline{\mathcal{R}}_i) \setminus \{\mathbf{0}\}$
    where $V_{nn}(\mathbf{a}) \leq 0$, $\mathbf{a}$ is a counterexample.
    Otherwise, Condition~\ref{cond_2} is satisfied.
\end{test}

We now move to the more complex scenario of verifying Condition~\ref{cond_3}. 

\subsection{Verifying Negative Lie Derivative (Condition~\ref{cond_3})}
In this section, we will show how Hyperlyve can verify Condition~\ref{cond_3} from Definition~\ref{def:lyapunov}. Recall that the Lie derivative of $V_{nn}$ along the system dynamics at the interior of a region $\mathcal{R}$ is given by
\[
\dot V_{nn}(\mathbf{x})=\nabla V_{nn}|_{\mathcal{R}}\cdot f(\mathbf{x}),
\]
where $\nabla V_{nn}|_{\mathcal{R}}$ represents the fixed gradient of $V_{nn}$ in the region.
From Proposition~\ref{constant_grad}, the gradient is constant within each region $\mathcal{R}_i$. Consequently, any variation in the Lie derivative over a region is due solely to the system dynamics $f(\mathbf{x})$. 

To guarantee that a counterexample is detected if it exists, we provide the following theorem. 
\begin{theorem}
        Suppose Assumption~\ref{assump: 1} holds. 
      There exists an $\mathbf{x}_0 \in \mathcal{B}$ such that \mbox{$\nabla V_{nn}|_{\mathcal{R}} \cdot f(\mathbf{x}_0) \geq 0$} if and only if there exists a $\mathbf{x} \in \mathcal{R}$ such that \mbox{$\max_{\mathbf{x} \in \mathcal{R}} \nabla V_{nn}|_{\mathcal{R}} \cdot f(\mathbf{x}) \geq 0$.}
    \label{thm: optimize}
\end{theorem}

\begin{proof}
    Suppose there is an $\mathbf{x}_0 \in \mathcal{B}$ such that \mbox{$\nabla V_{nn} |_{\mathcal{R}} \cdot f(\mathbf{x}_0) \geq 0$}, where $\mathbf{x}_0$ belongs to some $\mathcal{R}$. Then \mbox{$0 \leq \nabla V_{nn} |_{\mathcal{R}} \cdot f(\mathbf{x}_0) \leq \max_{\mathbf{x} \in \mathcal{R}} \nabla V_{nn} |_{\mathcal{R}} \cdot f(\mathbf{x})$}, completing the forward direction. 

    Now suppose that $\max_{\mathbf{x} \in \mathcal{R}}\nabla V_{nn} |_{\mathcal{R}} \cdot f(\mathbf{x}) \geq 0$. Then there must exist some $\mathbf{x}_0 \in \mathcal{R}$ such that \mbox{$\max_{\mathbf{x} \in \mathcal{R}} \nabla V_{nn} |_{\mathcal{R}} \cdot f(\mathbf{x}) \geq \nabla V_{nn} |_{\mathcal{R}} \cdot f(\mathbf{x}_0) \geq 0$} as desired. 
\end{proof}

Thus, Theorem~\ref{thm: optimize} yields the following test to be implemented to verify Condition~\ref{cond_3}.

\begin{test}[Negative Lie Derivative — Condition~\ref{cond_3}]
    \label{test:cond_3}
     Consider an activation region $\mathcal{R}$. Let  $\mathbf{x}^* = \arg\max_{\mathbf{x} \in \mathcal{R}} \nabla V_{nn} | _{\mathcal{R}} \cdot f(\mathbf{x})$.
    If $\nabla  V_{nn} |_{\mathcal{R}} \cdot f(\mathbf{x^*}) \geq 0$, return $\mathbf{x}^*$ as a counterexample.
    Otherwise, Condition~\ref{cond_3} is satisfied.
\end{test}

The next section develops Tests~\ref{test:cond_1}, ~\ref{test:cond_2}, and ~\ref{test:cond_3} into an algorithm that determines if a given neural network is indeed a valid neural Lyapunov function.

\begin{algorithm}[]
\caption{Hyperlyve: Neural Lyapunov Verification}\label{alg:nnv}
\begin{algorithmic}[1]
\Require Lyapunov network $V_{nn}(\mathbf{x})$, dynamics $f(\mathbf{x})$, compact domain $\mathcal{K}\subset \mathbb{R}^p$, with $\mathbf{x}\in \mathbb{R}^p$
\Ensure Pass/fail with counterexamples $S$

\State $S \leftarrow \emptyset$
\If{$V_{nn}(\mathbf{0}) \neq 0$} \Comment{Test \ref{test:cond_1}}
    \State $S \leftarrow S \cup \{\mathbf{0}\}$
\EndIf

\State $\mathcal{A} \leftarrow \textsc{HyperplaneArrangement}(V_{nn})$ \Comment{Def. \ref{def: arrangement}}
\State $\{\mathcal{B}\} \leftarrow \textsc{EnumerateRegions}(\mathcal{A},\ \mathcal{K})$ \label{alg:nnv:line:enumerate}
\Comment{Xu et al.~\cite{xu2022traversing}}

\For{each $\mathcal{R}_i \in \mathcal{B}$}
    \For{each vertex $\mathbf{x} \in \bigcup_i \mathrm{vert}(\mathcal{R}_i)\setminus \{\mathbf{0}\}$}
        \If{$V_{nn}(\mathbf{x}) \leq 0$} \Comment{Test \ref{test:cond_2} }
            \State $S \leftarrow S \cup \mathbf{x}$
        \EndIf
    \EndFor
\EndFor

\For{each $\mathcal{R}_i \in \mathcal{B}$} \Comment{Test \ref{test:cond_3}}
    \State $\mathbf{c}_i \leftarrow \textsc{ChebyshevCenter}(\mathcal{R}_i)$ 
    \State $\mathbf{x}^* \leftarrow \arg\max_{\mathbf{x} \in \mathcal{R}_i} \
    \nabla V_{nn}(\mathbf{c}_{i}) \cdot f(\mathbf{x})$\label{alg:line:argmaxfx}\Comment{SHGO~\cite{endres2018simplicial}}
    \If{$\nabla V_{nn}(\mathbf{c}_{i}) \cdot f(\mathbf{x}^*) \geq 0$}
        \State  $S \leftarrow S \cup \mathbf{x}^*$
    \EndIf
\EndFor

\If{$S = \emptyset$}
    \State  \Return \textsc{Verified}
\Else
    \State \Return $S$
\EndIf
    
\end{algorithmic}
\end{algorithm}

\section{Verification Algorithm}
\label{section: our_alg}

Our solution to Problem~\ref{prob: 1} is expressed in Algorithm~\ref{alg:nnv}, which requires a candidate Lyapunov network $V_{nn}(\mathbf{x})$, dynamics $f(\mathbf{x})$, and compact domain $\mathcal{K}\subset \mathbb{R}^p$, with $\mathbf{x}\in \mathbb{R}^p$. The algorithm returns \textsc{Verified} if and only if Condition \ref{cond_1}, \ref{cond_2}, and \ref{cond_3} from Definition~\ref{def:lyapunov} hold; otherwise, a set $S$ of counterexamples is returned.

\begin{figure}[t]
    \centering
    \begin{subfigure}[b]{0.45\columnwidth}
        \centering
        \begin{tikzpicture}[scale=0.8]
            \draw[thick] (-2,1) -- (2,1) node[right] {$H_a$};    
            \draw[thick] (-2,-1) -- (2,-1) node[right] {$H_c$};  
            \draw[thick] (-1.5,-2) -- (1.5,2) node[above] {$H_b$}; 
            
            \filldraw (0.75,1) circle (2pt) node[above left] {$X$}; 
            \filldraw (-0.75,-1) circle (2pt) node[below right] {$Y$}; 
            
            \node at (-1.8,1.8) {$\mathbb{R}^2$};
        \end{tikzpicture}
        \caption{Arrangement $\mathcal{A}$}
        \label{fig:arrangement_example}
    \end{subfigure}
    \hfill
    \begin{subfigure}[b]{0.45\columnwidth}
        \centering
        \begin{tikzpicture}[scale=1.2]
            \node (R2) at (0,0) {$\mathbb{R}^2$};  
            \node (a) at (-1,1) {$H_a$};             
            \node (b) at (0,1) {$H_b$};
            \node (c) at (1,1) {$H_c$};
            \node (X) at (-0.5,2) {$X$};           
            \node (Y) at (0.5,2) {$Y$};
            
            \draw (R2) -- (a);
            \draw (R2) -- (b);
            \draw (R2) -- (c);
            \draw (a) -- (X);
            \draw (b) -- (X);
            \draw (b) -- (Y);
            \draw (c) -- (Y);
        \end{tikzpicture}
        \caption{Intersection Poset $L(\mathcal{A})$}
        \label{fig:hasse}
    \end{subfigure}
    
    \caption{A hyperplane arrangement and its Hasse diagram (of the intersection poset). The diagram is built by dimension, placing highest-dimensional elements at the bottom. Since $\dim(\mathbb{R}^2)=2$, it appears first. Each hyperplane has dimension 1 and lies in $\mathbb{R}^2$, and edges indicate containment. Finally, the intersection points (dimension 0) are added: $X \subset H_a, H_b$ gives edges from $X$ to $H_a, H_b$, and $Y \subset H_b, H_c$ gives edges from $Y$ to $H_b, H_c$.}
    \label{fig:polynomial_example}
\end{figure}
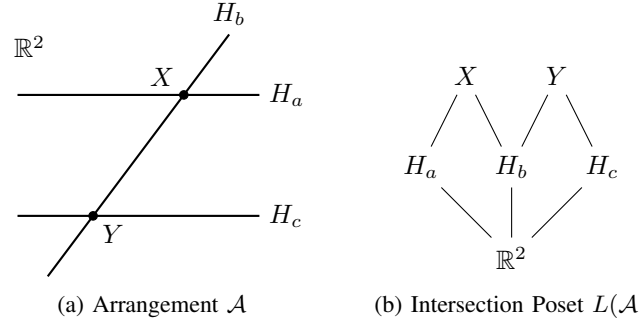

Our algorithm is somewhat unique in that instead of returning only the first identified counterexample, we return a representative set $S$ that contains at least one counterexample for every invalid region.
Since counterexamples are tied to specific polytopes spread throughout $\mathcal{B}$, $S$ is likely to contain a more representative set of Lyapunov network violations, which can assist in retraining.

 We implement the neural networks and dynamics in PyTorch~\cite{ansel2024pytorch}.
 Enumerating the affine regions in Line~\ref{alg:nnv:line:enumerate} is done according to the single-layer BFS-based polytope traversing of~\cite{xu2022traversing}.
 
 \paragraph{Nonlinear Optimization} Handling Test~\ref{test:cond_3} requires solving the non-convex optimization problem in Line~\ref{alg:line:argmaxfx}. While global optimization of non-convex functions is NP-hard in general~\cite{vavasis1995complexity}, we ensure reliable practical performance and asymptotic soundness by using the Simplicial Homology Global Optimization (SHGO) algorithm~\cite{endres2018simplicial} implemented in SciPy~\cite{2020SciPy-NMeth}, coupled with Sequential Least Squares Programming (SLSQP)~\cite{kraft1988software} as the local solver.
 SHGO constructs a simplicial complex from domain samples, partitions it into locally convex subdomains, and passes these to local solvers. SHGO guarantees convergence to global optima as the sample density approaches infinity~\cite{endres2018simplicial}.
To satisfy SHGO's sampling requirement across the polytope $\mathcal{R}_i$, we anchor the simplical complex boundary using the verticies of $\mathcal{R}_i$ and sample its interior using a Hit-and-Run Markov chain~\cite{smith1984efficient}. Hit-and-Run is proven to mix rapidly in polynomial time from any interior point, including the Chebyshev center~\cite{lovasz2006hit}.
We pass the linear constraints $Ax\leq b$ defining $\mathcal{R}_i$ directly to SLSQP, while nonlinear constraints exclude a ball of radius $10^{-6}$ around the origin for numerical stability; thus, the optimizer only considers optima relevant to Test~\ref{test:cond_3}.
Finally, we compute the Jacobian of the dynamics via PyTorch Autograd and provide it to SLSQP, avoiding finite-difference approximations and accelerating local convergence.
Future work will explore replacing SHGO with formal verification methods.

\subsection{Computational Complexity}
We now explore the complexity of our algorithm, which is dependent on the number of regions produced by a hyperplane arrangement. Thus, the following definitions act as supporting information for the next theoretical result. These definitions are mainly adapted from \cite{stanley2007introduction}. Recall that $r(\mathcal{A})$ denotes the number of regions produced from a hyperplane arrangement $\mathcal{A}$, and $r_{\mathcal{B}}(\mathcal{A})$ denotes the number of polytopes contained in $\mathcal{B}$.
\begin{definition}[Intersection Poset]
    The intersection poset $L(\mathcal{A})$ is defined as the partially ordered set of nonempty intersections of a given hyperplane arrangement $\mathcal{A}$. It is a set of possible intersections from all possible combinations of elements in $\mathcal{A}$. It is ordered by reverse inclusion, i.e., $B \subseteq A \implies A \leq B$ for sets $A$ and $B$. 
    \label{def: poset}
\end{definition}

\begin{definition}[M\"obius Function]
Given a hyperplane arrangement $\mathcal{A}$ and intersection poset $L(\mathcal{A})$, the M\"obius Function $\mu: L(\mathcal{A}) \to \mathbb{Z}$ has the following properties:
\begin{enumerate}[label=(\roman*)] 
    \item $\mu(\hat{0}) = 1$
    \item $\mu(y) = - \sum_{x < y}\mu(x)$, for all $x, y \in L(\mathcal{A}), y \neq \hat{0}$. 
\end{enumerate}
\label{def:mobius}
\end{definition}

\begin{table}
\centering
\caption{The first three columns are in the following order: number of hidden neurons, dimension of the state space, and number of polytopes in $\mathcal{B}$. The last column is the upper bound calculated from Lemma~\ref{lem_upper_bound}. The region of interest in each row was defined as $\mathcal{B} = [-10, 10]^p$. }
\label{tab:cubic-decrease}
\begin{tabular}{@{}ccrr@{}}
    \toprule
    $n$ & $p$ & $r_{\mathcal{B}}(\mathcal{A})$ & Upper Bound  \\
    \midrule
      4 &  2 &    4 &    7 \\
      6 &  3 &    8 &   15 \\
      8 &  4 &   16 &   31 \\
     10 &  5 &   32 &   63 \\
     12 &  6 &   64 &  127 \\
     14 &  7 &  128 &  255 \\
     16 &  8 &  256 &  511 \\
     18 &  9 &  512 & 1023 \\
     20 & 10 & 1024 & 2047 \\
    \bottomrule
    \end{tabular}
\end{table}

\begin{definition}[Characteristic Polynomial]
    Given a hyperplane arrangement $\mathcal{A}$ and its corresponding intersection poset $L(\mathcal{A})$, the characteristic polynomial is defined as
    \begin{equation}
        \chi_{\mathcal{A}}(t) = \sum_{x \in L(\mathcal{A})} \mu(x)t^{\dim(x)} 
    \end{equation}
    where $\dim(x)$ is the dimension of an element of the poset. 
    \label{def:char_poly}
\end{definition}

To illustrate the computation of the characteristic polynomial, we provide a simple example. Consider the hyperplane arrangement given in Figure~\ref{fig:arrangement_example}. The hyperplanes are $H_a, H_b$ and $H_c$ with intersections $X$ and $Y$. The intersection poset is $L(\mathcal{A}) = \{\mathbb{R}^2, H_a, H_b, H_c, X, Y\}$. As $\mathbb{R}^2$ is the largest space in the poset, reverse inclusion causes it to be the minimum (i.e., all other poset elements are contained in $\mathbb{R}^2$). 

To find the characteristic polynomial, we must first find the value of the M\"obius function at each element in the intersection poset. Since $\mathbb{R}^2$ is the minimum, $\mu(\mathbb{R}^2) = 1$. For each hyperplane $H_i$, where $i \in \{a, b, c\}$, the only element less than it is $\mathbb{R}^2$. Thus, $\mu(H_a) = \mu(H_b) = \mu(H_c) = -\mu(\mathbb{R}^2) = -1$. Since $X$ is the intersection of $H_a$ and $H_b$, we know that $X \subseteq H_a$ and $X \subseteq H_b$, and by reverse inclusion $X \geq H_a, H_b$. Therefore, $\mu(X) = -(\mu(H_a) + \mu(H_b) + \mu(\mathbb{R}^2)) = 1$. Likewise, $\mu(Y) = 1$. Note that the M\"obius function at each element in the intersection poset can be computed recursively by summing up the values of the M\"obius function values of the elements connected from below and flipping the sign in the Hasse diagram showed in Figure~\ref{fig:hasse}.

Referring to the definition of the characteristic polynomial, we then get that $\chi_{\mathcal{A}}(t) = 1 \cdot t^2 + (-1 -1 -1) \cdot t^1 + (1 + 1) \cdot t^0 = t^2 - 3t + 2$. This polynomial will be useful when we count the regions produced by a hyperplane arrangement. 

We lastly describe how the number of regions grows with the dimension of the state space $p$. 
\begin{lemma}
    Consider a candidate neural network Lyapunov function $V_{nn}(\mathbf{x})$ and a region of interest $\mathcal{B}$. Then the number of regions is bounded:
    \begin{equation}
        r_{\mathcal{B}}(\mathcal{A}) \leq (-1)^p \chi_{\mathcal{A}}(-1). 
    \end{equation}
    \label{lem_upper_bound}
\end{lemma}
\begin{proof}
    It follows from Zaslavsky's Theorem \cite{stanley2007introduction} \cite{zaslavsky1997facing}  that 
    $$r(\mathcal{A}) = (-1)^p \chi_{\mathcal{A}}(-1), $$
    where $r(\mathcal{A})$ represents the number of regions produced by a hyperplane arrangement $\mathcal{A}$ and $\chi_{\mathcal{A}}(-1)$ is the corresponding characteristic polynomial evaluated at $-1$. However, our analysis is only over a compact region of interest $\mathcal{B}$, and there could potentially be intersections outside of $\mathcal{B}$. Thus, $r_{\mathcal{B}}(\mathcal{A}) \leq r(\mathcal{A})$.
\end{proof}
Recall the hyperplane arrangement from Figure~\ref{fig:arrangement_example} with characteristic polynomial $\chi_{\mathcal{A}}(t) = t^2 - 3t + 2$. Following Zaslavsky's Theorem, we compute $r(\mathcal{A}) = (-1)^2 (1 + 3 + 2) = 6$. Upon visual inspection of Figure~\ref{fig:arrangement_example}, we confirm that there are exactly $6$ regions. 

\begin{proposition}
    The time complexity for Algorithm~\ref{alg:nnv} is $\mathcal{O}( (-2)^p \chi_{\mathcal{A}}(-1))$.
    \label{prop: time_complexity}
\end{proposition}
\begin{proof}
The dominating portion of the algorithm in terms of time complexity is lines 14-22. Iterating through the number of regions scales by the upper bound found in Lemma~\ref{lem_upper_bound}. For each partition, we solve an optimization problem on line 17, which scales by $2^p$, where $p$ is the dimension of the system. Thus, the overall time complexity is $\mathcal{O}( (2)^p (-1)^p \chi_{\mathcal{A}}(-1)) = \mathcal{O}( (-2)^p \chi_{\mathcal{A}}(-1))$.
\end{proof}
As shown in Proposition~\ref{prop: time_complexity}, the time complexity is dependent on the number of total regions a hyperplane arrangement generates. To illustrate this result, the upper bound from Lemma~\ref{lem_upper_bound} for the neural network $V_{nn}(x)=\sum_i^p \sigma(x_i)+\sigma(-x_i)=\sum_i^p |x_i|=||\mathbf{x}||_1$ is given in Table~\ref{tab:cubic-decrease} with a variable number of neurons and state dimension. The region of interest was fixed $[-10, 10]^p$. Notice that the number of regions in $\mathcal{B}$ is much less than the upper bound. In this situation, Proposition~\ref{prop: time_complexity} yields a more conservative bound on the time complexity of our algorithm, but a larger $\mathcal{B}$ would lead to a tighter bound. 

\begin{table}
    \centering
    \caption{Lyapunov verification results for Hyperlyve and $\alpha$-$\beta$-CROWN across network sizes and dynamical systems. The columns denote the number of neurons in the network,  state space dimension, number of regions in $\mathcal{B}$, number of counterexamples detected by Hyperlyve, Hyperlyve time taken, and CROWN time taken. Timeout is two hours.}
    \label{tab:results}
\begin{threeparttable}
\begin{tabular}{@{}cclccc@{}}
    \toprule
    $n$ & $p$ & $r_{\mathcal{B}}(\mathcal{A})$ & $|S|$ & \multicolumn{1}{c}{Hyperlyve (s)} & \multicolumn{1}{c}{$\alpha$-$\beta$-CROWN (s)} \\
    \midrule
    \multicolumn{6}{@{}l}{\textit{Negative Cubic \eqref{eq:cubic_decrease}}} \\[1pt]
      4 &  2 &    4 & 0 &    1.2 &  4.7 \\
      6 &  3 &    8 & 0 &    3.7 &   11.8 \\
      8 &  4 &   16 & 0 &    3.8 &   13.7 \\
     10 &  5 &   32 & 0 &    4.3 &   18.1 \\
     12 &  6 &   64 & 0 &    5.5 &   29.2 \\
     14 &  7 &  128 & 0 &   12.2 &   61.0 \\
     16 &  8 &  256 & 0 &   17.6 &  130 \\
     18 &  9 &  512 & 0 &   37.8 &  313 \\
     20 & 10 & 1024 & 0 &  115 &    740 \\
    \midrule
    \multicolumn{6}{@{}l}{\textit{Bilinear Oscillator ~\eqref{eq:bilinear_oscillator}}} \\[1pt]
      10 &  2 &    54\tnote{a} & 20 &    1.8 &   timeout \\
      10 &  2 &    55 & 0  &    2.0 &   9.8 \\
     \midrule
    \multicolumn{6}{@{}l}{\textit{Coupled Bilinear Oscillators \eqref{eq:4d_oscillators}}} \\[1pt]
     50 & 4 & 14,474    &   3 &   390 & \multicolumn{1}{c}{timeout} \\
     80 & 4 & 1,021,050 & 574 &  1426 &          1458             \\
     80 & 4 & 876,655   & 0 &  1519 & \multicolumn{1}{c}{timeout} \\
    \bottomrule
\end{tabular}\begin{tablenotes}\footnotesize

      \item[a] \textit{This is the network shown in Figure \ref{fig:cells}.}
      \end{tablenotes}
\end{threeparttable}
\end{table}
\section{SIMULATIONS} 
\label{section: simulations}

We evaluate our algorithm using  three different system dynamics. First, we use a negative cubic function:
\begin{equation}
\label{eq:cubic_decrease}
\dot{x}_i = -x_i^3,
\end{equation}
for $i \in [10]$ and $\mathcal{B} =[-10, 10]^i$.
Second, we use a bilinear oscillator:
\begin{equation} \label{eq:bilinear_oscillator}
\begin{split}
    \dot{x}_1 &= -x_1 + x_1 x_2 \\
    \dot{x}_2 &= -x_2 - x_1^2
\end{split}
\end{equation}
with $\mathcal{B} = [-4, 4]^2$.
Finally, we use a coupled pair of bilinear oscillators:
\begin{equation} \label{eq:4d_oscillators}
    \begin{split}
    \dot{x}_1 &= -x_1 + x_1 x_2\\
    \dot{x}_2 &= -x_2 - x_1^2+ \varepsilon x_3\\
    \dot{x}_3 &= -x_3 + x_3 x_4\\
    \dot{x}_4 &= -x_4 - x_3^2 + \varepsilon x_1, 
    \end{split}
\end{equation} 
with $\varepsilon = 0.1$ and $\mathcal{B} = [-2, 2]^4$.

The benchmark algorithm we use is the $\alpha$-$\beta$-CROWN branch-and-bound complete verifier~\cite{wang2021beta}. Experiments were conducted on CentOS Linux with an AMD EPYC 7662 64-Core Processor and NVIDIA A100 40GB GPU.
Similarly to~\cite{yang2024lyapunov}, we omit a hypercube hole of $0.1\%$ of the region of interest size from Tests \ref{test:cond_2} and \ref{test:cond_3} for float32 numerical stability. Note that $\alpha$-$\beta$-CROWN requires hyperrectangular bounds. To avoid the discontinuity at the origin, we constructed each $p$-dimension neural Lyapunov verification as a set of $4p$ problems for $\alpha$-$\beta$-CROWN to solve.

Table \ref{tab:results} contains comparative results for the three different system dynamics and different network sizes. The columns are in the following order: number of neurons in the candidate, state space dimension, number of regions contained in $\mathcal{B}$, number of counterexamples detected by Hyperlyve, time taken by Hyperlyve, and time taken by CROWN (in seconds). Furthermore, $|S|$ implies that the candidate is a valid Lyapunov function. Note that our algorithm outperforms $\alpha$-$\beta$-CROWN in every instance.


At the 2-hour per-network timeout, $\alpha$-$\beta$-CROWN had verified Condition \ref{cond_2} for the entire 50-neuron bilinear oscillator network and Condition \ref{cond_3} for 7-of-8 partitions of the network. The final partition failed to converge after branch-and-bound explored 136-million nodes. 
For this case, our algorithm found 3 counterexamples in 390 seconds. 
For the 80-neuron network that our algorithm verified in 1519 seconds, $\alpha$-$\beta$-CROWN neither proved nor disproved any properties before timeout.

In Figure~\ref{fig:cells}, we show a plot of the candidate neural Lyapunov function for the bilinear oscillator (the first row in the second section of Table~\ref{tab:results}) with the induced regions $\mathcal{R}_i$ overlayed in different colors. The locations of the 20 counterexamples detected by Hyperlyve are depicted in red. 
Note that the set of counterexamples in Figure~\ref{fig:cells} is distributed throughout the region of interest, and could be particularly useful for retraining the neural network to get closer to a valid Lyapunov function. 

\begin{figure}
    \centering
    \includegraphics[width=\linewidth]{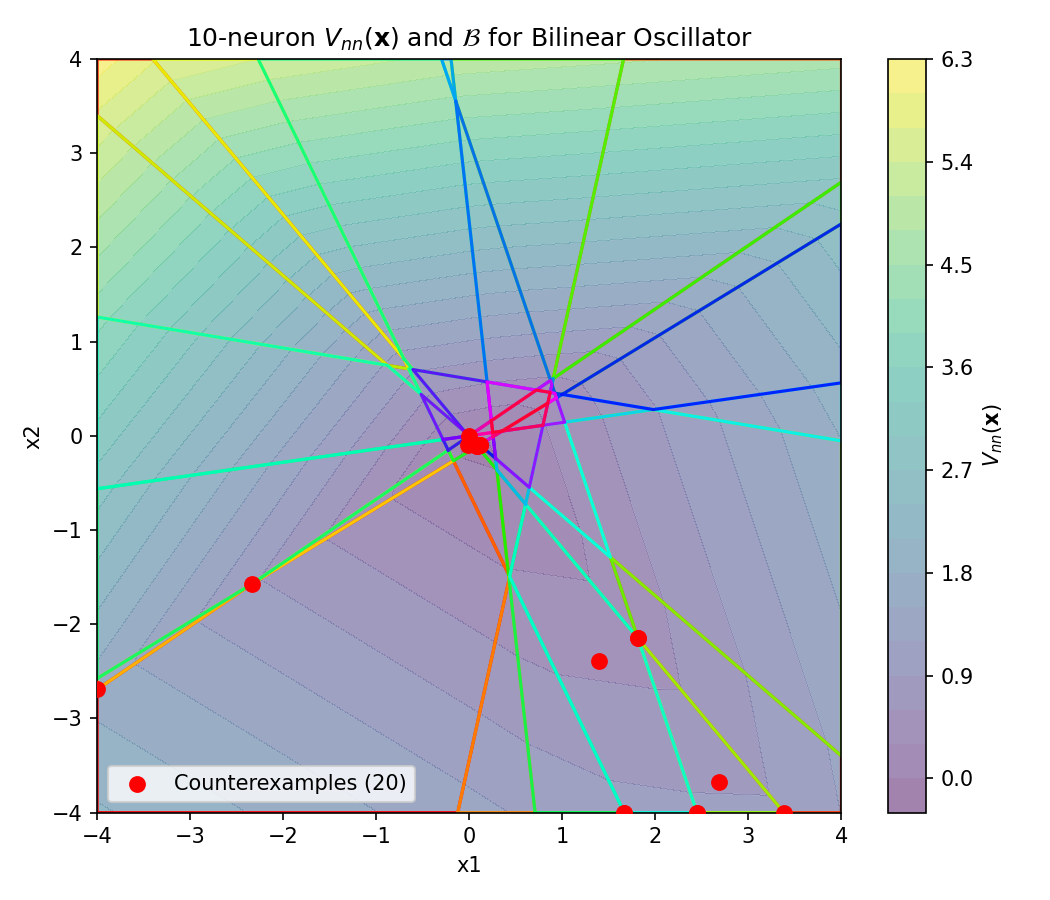}
    \caption{A candidate 10-neuron $V_{nn}$ for bilinear oscillator dynamics: $\dot{x}_1 = -x_1 + x_1x_2$, $\dot{x}_2=-x_1^2$. There are 54 regions outlined in distinct colors (note: very small regions do not render in the image). Algorithm \ref{alg:nnv} identifies 2 vertices that failed Test \ref{test:cond_2} and 8 regions that failed Test \ref{test:cond_3}.}
    \label{fig:cells}
\end{figure}

\section{CONCLUSION}
\label{section: conclusion}
We developed a novel method for verifying ReLU neural Lyapunov functions by leveraging the geometric properties of hyperplane arrangements. We proved that the positivity condition can be verified by checking the vertices of each region, and also reduced the decrease condition to a sequence of optimization problems. We further proved results on the computational complexity of our algorithm, and demonstrated the effectiveness of Hyperlyve through comparison with state-of-the-art verification methods. For future work, we plan to extend this method to deep neural networks as well as neural network controllers. Another direction for future work is to leverage the produced counterexamples from Hyperlyve to retrain neural Lyapunov functions, guiding them towards valid Lyapunov functions.




\bibliographystyle{IEEEtran} 
\bibliography{mybib} 

\end{document}